\documentclass[conference,letterpaper]{IEEEtran}

\usepackage[utf8]{inputenc} 
\usepackage[T1]{fontenc}
\usepackage{url}
\usepackage{ifthen}
\usepackage{cite}
\usepackage{comment}
\usepackage[cmex10]{amsmath} 

\usepackage{graphicx,xcolor,amssymb,amsthm,cite,url}
\usepackage{algorithm}
\usepackage{algpseudocode}
\usepackage{subcaption}

\graphicspath{{figs/}}
\DeclareMathOperator*{\argmin}{arg\,min}
\DeclareMathOperator*{\argmax}{arg\,max}

\newtheorem{prop}{Proposition}

\newcommand{\xbf}{\mathbf{x}}
\newcommand{\ybf}{\mathbf{y}}
\newcommand{\wbf}{\mathbf{w}}
\newcommand{\vbf}{\mathbf{v}}
\newcommand{\sbf}{\mathbf{s}}
\newcommand{\ubf}{\mathbf{u}}
\newcommand{\ebf}{\mathbf{e}}
\newcommand{\zbf}{\mathbf{z}}

\newcommand{\Abf}{\mathbf{A}}
\newcommand{\Jbf}{\mathbf{J}}
\newcommand{\Ibf}{\mathbf{I}}
\newcommand{\Dbf}{\mathbf{D}}
\newcommand{\Bbf}{\mathbf{B}}
\newcommand{\Hbf}{\mathbf{H}}

\newcommand{\row}{\mathsf{R}}
\newcommand{\col}{\mathsf{C}}
\newcommand{\Null}{\text{Null}}

\IEEEoverridecommandlockouts

\begin{document}
\title{Local Convergence of an AMP Variant \\ to the LASSO Solution in Finite Dimensions}


\author{%
\IEEEauthorblockN{Yanting Ma\IEEEauthorrefmark{1},
Min Kang\IEEEauthorrefmark{2},
Jack W. Silverstein\IEEEauthorrefmark{2},
Dror Baron\IEEEauthorrefmark{3}}
\IEEEauthorblockA{\IEEEauthorrefmark{1}%
                    Mitsubishi Electric Research Laboratories,
                    yma@merl.com}
  \IEEEauthorblockA{\IEEEauthorrefmark{2}%
                    Department of Mathematics,
                    North Carolina State University,
                    \{mkang2, Jack\}@ncsu.edu}
  \IEEEauthorblockA{\IEEEauthorrefmark{3}%
                    Department of Electrical and Computer Engineering,
                    North Carolina State University,
                    barondror@ncsu.edu}
\thanks{D. Baron was supported in part by NSF EECS 1611112.
This work was completed while Y. Ma was with North Carolina State University.}}

\maketitle
\thispagestyle{empty}

\begin{abstract}
A common sparse linear regression formulation is $\ell_1$ regularized least squares, which is also known as least absolute shrinkage and selection operator (LASSO).
Approximate message passing (AMP) has been proved to asymptotically achieve the LASSO solution when the regression matrix has independent and identically distributed (i.i.d.) Gaussian entries in the sense that the averaged per-coordinate $\ell_2$ distance between the AMP iterates and LASSO solution vanishes as the signal dimension goes to infinity before the iteration number. 
However, in finite dimensional settings, characterization of AMP iterates in the limit of large iteration number has not been established. 
In this work, we propose an AMP variant by including a parameter that depends on the largest singular value of the regression matrix. The proposed algorithm can also be considered as a primal dual hybrid gradient algorithm with adaptive stepsizes.  
We show that whenever the AMP variant converges, it converges to the LASSO solution for arbitrary finite dimensional regression matrices. 
Moreover, we show that our AMP variant is locally stable around the LASSO solution 
under the condition that 
the LASSO solution is unique and that the regression matrix is drawn from a continuous distribution. 
Our local stability result implies that 
when the regression matrix is large and has i.i.d. random entries, the original AMP, which is a special case of the proposed AMP variant, is locally stable around the LASSO solution. 
\end{abstract}

\section{Introduction}
\label{sec:intro}

Least absolute shrinkage and selection operator (LASSO) is a common formulation for sparse linear regression, which is defined as the optimization problem:
\begin{equation}
\xbf^*\in \argmin_{\xbf\in\mathbb{R}^N} \{F(\xbf):=\frac{1}{2} \| \ybf - \Abf \xbf \|_2^2 + \gamma \| \xbf \|_1\},
\label{eq:lasso}
\end{equation}
where $\Abf\in\mathbb{R}^{n\times N}$ is the regression matrix, $\mathbf{y}\in\mathbb{R}^n$ is the data vector, $\gamma>0$ is the regularization parameter, and $\|\cdot\|_p$, for $p=1,2$, denotes the $\ell_p$ norm. 
While numerous standard convex optimization algorithms  such as the class of proximal gradient methods~\cite{combettes2005signal,Nesterov2007,Daubechies.etal2004,BeckTeboulle2009}, alternating direction method of multipliers (ADMM)~\cite{Boyd.etal2011}, and primal dual hybrid gradient (PDHG)~\cite{Esser2010,ChambollePock2011,HeYuan2012} can be used to solve \eqref{eq:lasso}, it is of both theoretical and practical interest to study approximate message passing (AMP) for solving \eqref{eq:lasso}, since AMP was initially introduced by Donoho et al.~\cite{Donoho.etal2009} as a LASSO solver and usually enjoys fast empirical convergence when it converges. 

Existing theoretical convergence analyses of standard optimization algorithms and
of AMP are considered in different problem settings. Specifically, the quantity of interest to optimization algorithms is usually $\lim_{t\to\infty}\|\xbf^t - \xbf^*\|_2$, where $\xbf^t$ is the estimate at the $t^{th}$ iteration of an iterative algorithm, for any fixed and finite $n$ and $N$. It is usually assumed in the AMP framework that the data vector $\ybf$ is generated according to a linear system, $\ybf = \Abf\xbf_0 + \wbf$, with some underlying ground-truth $\xbf_0$ and noise $\wbf$. Under these assumptions, the analysis of AMP shows that when $\Abf$ has independent and identically distributed (i.i.d.) Gaussian entries, the quantity $\lim_{N\to\infty}\frac{1}{N}\|\xbf^t - \xbf_0\|_2$  with $\frac{n}{N}\to c\in (0,\infty)$ converges to a deterministic number predicted by a scalar recursion referred to as state evolution with probability one~\cite{BayatiMontanari2011}; this is later extended to a large deviation result~\cite{RushV2018}. 
For the class of right-orthogonal invariant random $\Abf$, vector-AMP~\cite{rangan2019vector} (see also orthogonal AMP~\cite{ma2017orthogonal}) also enjoys a rigorous state evolution analysis as $n,N\to\infty$ and $\frac{n}{N}\to c$.
For the LASSO problem, Bayati and Montanari~\cite{BayatiMontanari2012} have proven the convergence of AMP iterates to the LASSO solution $\xbf^*$ in the sense that
$\lim_{t\to\infty} \lim_{N\to\infty} \frac{1}{N} \|\xbf^t - \xbf^*  \|_2^2 =0$ with probability one, which has also been extended to a large deviation result in recent work~\cite{Rush2020}. However, this large deviation result only holds for $t=O\left(\frac{\log N}{\log\log N}\right)$.\footnote{The big O notation $O(\cdot)$ means that there exists an $N_0\in\mathbb{N}$ and a positive real number $B$ such that $t(N)\leq B \frac{\log N}{\log\log N}$ for all $N\geq N_0$.} Therefore, the convergence of AMP for finite $N$ as $t\to\infty$ is still unknown.
In fact, using matrices with i.i.d. Gaussian entries and following the calibration method proposed in \cite{BayatiMontanari2012} for choosing the threshold of the soft-thresholding function at each AMP iteration, we performed 2000 trials of Monte Carlo simulations with $N=2000, n=1000$, and AMP never converged to the LASSO solution in terms of $\ell_2$ error.

The connection between AMP and standard convex optimization algorithms has enabled the design of AMP variants that have convergence guarantees for more practical settings such as non-Gaussian finite dimensional matrices. Most such results have been developed in a more general algorithmic framework known as generalized AMP (GAMP)~\cite{Rangan2011}. In the context of solving optimization problems, GAMP considers objective functions of the form $\sum_{i=1}^{n} g([\Abf\xbf]_i)+\sum_{i=1}^N f(x_i)$. Consider now that $\Abf$ is arbitrary and finite dimensional. When both $g$ and $f$ are quadratic functions,  damped GAMP~\cite{rangan2019convergence}, which defines the current iterate as a convex combination of the current estimate and the iterate from the previous iteration, has global convergence guarantees. When $g$ and $f$ are strictly convex and twice continuously differentiable, assuming that the derivatives of the nonlinear functions used in each GAMP iteration are bounded within the open interval $(0,1)$, damped GAMP with fixed stepsize~\cite{rangan2019convergence} is proved to be locally stable around the equilibrium point, and ADMM-GAMP~\cite{Rangan.etal2017}, which combines an ADMM inner loop within each GAMP iteration, is guaranteed to achieve global convergence. 

The interpretation of AMP as a PDHG algorithm was first mentioned by Rangan et al.~\cite{rangan2019convergence}, which has inspired the current work.
The differences between our work and the local stability analysis in~\cite{rangan2019convergence} are as follows:\\
$\bullet$ 
The objective function of LASSO is non-differentiable, hence is not covered in~\cite{rangan2019convergence}.\\
$\bullet$ 
Instead of damping the iterates while keeping the stepsizes fixed as in \cite{rangan2019convergence}, we do not make changes to updates of the iterates but design a stepsize updating schedule based on AMP.\\  
$\bullet$ 
The result in \cite{rangan2019convergence} holds only when the stepsize is fixed over all iterations, which loses the main advantage of AMP over standard optimization algorithms for fast convergence, whereas our result allows keeping the structure of the stepsize updating schedule of AMP. Numerical results in Section \ref{sec:numerical}  
show that the number of iterations required for our algorithm to converge to the LASSO solution is orders of magnitude smaller than widely used optimization algorithms.

\section{Proposed Algorithm}
\label{sec:algo}

Let $g:\mathbb{R}^n\to\mathbb{R}$ and $f:\mathbb{R}^N\to\mathbb{R}$ be defined as
\begin{equation}
g(\ubf):=\frac{1}{2}\|\ybf-\ubf\|_2^2 \quad \text{and}\quad f(\xbf):=\gamma \|\xbf\|_1,
\label{eq:def_fg}
\end{equation}
respectively.
The idea of the class of PDHG algorithms is to write the minimization problem $\inf_{\xbf\in\mathbb{R}^N} \left\{ g(\Abf\xbf) + f(\xbf)\right\}$ as a saddle-point problem using the fact that the function $g$ defined above is convex, closed, and proper, thus $g=(g^{*})^*$~\cite{Rockafellar1970}, where $g^*$ is the convex-conjugate of $g$ defined as
\begin{equation}
g^*(\sbf):=\sup_{\ubf\in\mathbb{R}^n} \left\{ \langle \sbf,\ubf\rangle - g(\ubf)\right\}.
\label{eq:def_conv_conj}
\end{equation}
With $g(\Abf\xbf)=(g^*)^*(\Abf\xbf)=\sup_{\sbf\in\mathbb{R}^n} \left\{ \langle \Abf\xbf, \sbf \rangle - g^*(\sbf)\right\}$,
we obtain the saddle-point problem
\begin{equation}
\inf_{\xbf\in\mathbb{R}^N}\sup_{\sbf\in\mathbb{R}^n}  F(\sbf,\xbf),
\label{eq:saddle}
\end{equation}
where $F(\sbf,\xbf):=\langle \sbf, \Abf\xbf \rangle - g^*(\sbf) + f(\xbf)$.
PDHG solves \eqref{eq:saddle} by alternating between the estimation of $\sbf$ and $\xbf$ as $\sbf^{t+1}=\argmax_{\sbf\in\mathbb{R}^n} \left\{F(\sbf,\xbf^t)+\frac{1}{2\tau_s^t}\|\sbf - \sbf^t\|_2^2\right\}$ and $\xbf^{t+1}=\argmin_{\xbf\in\mathbb{R}^N} \left\{ F(\sbf^{t+1},\xbf)+\frac{1}{2\tau_x^t}\|\xbf - \xbf^t\|_2^2\right\}$, respectively, which is equivalent to
\begin{equation}
{\footnotesize
\begin{split}
\sbf^{t+1} &=\argmin_{\sbf\in\mathbb{R}^n} \left\{ g^*(\sbf) + \frac{1}{2\tau_s^t}\|\sbf-(\sbf^t + \tau_s^t \Abf\xbf^t)\|_2^2\right\},\\
\xbf^{t+1} &=\argmin_{\xbf\in\mathbb{R}^N} \left\{ f(\xbf)+ \frac{1}{2\tau_x^t}\|\xbf-(\xbf^t - \tau_x^t \Abf^T \sbf^{t+1})\|_2^2\right\}.
\end{split}
}
\label{eq:primal_dual_general}
\end{equation}
In the above, the stepsizes $\tau_x^t$ and $\tau_s^t$ can stay constant for all iterations or be updated at every iteration. One feature of PDHG algorithms is that each equilibrium point is a saddle-point of \eqref{eq:saddle}. This can be explained as follows. Let $(\widehat{\xbf},\widehat{\sbf},\widehat{\tau}_x,\widehat{\tau}_s)$ be an equilibrium point of the algorithm \eqref{eq:primal_dual_general}, then
\begin{equation*}
\begin{split}
0 &\in \partial_\sbf \Big( g^*(\sbf) + \frac{1}{2\widehat{\tau}_s}\|\sbf- (\widehat{\sbf} + \widehat{\tau}_s \Abf \widehat{\xbf})\|_2^2\Big)\Big\vert_{\sbf=\widehat{\sbf}}\,,\\
0 &\in \partial_\xbf \Big( f(\xbf) + \frac{1}{2\widehat{\tau}_x}\|\xbf- (\widehat{\xbf} - \widehat{\tau}_x \Abf^T \widehat{\sbf})\|_2^2\Big)\Big\vert_{\xbf = \widehat{\xbf}}\,,
\end{split}
\end{equation*} 
where $\partial_\ubf$ denotes sub-differential with respect to $\ubf$.
The above implies that $\Abf\widehat{\xbf}\in\partial g^*(\widehat{\sbf})$ and $-\Abf^T \widehat{\sbf}\in\partial f(\widehat{\xbf})$, which is the necessary and sufficient condition for $(\widehat{\xbf},\widehat{\sbf})$ to be a saddle-point of \eqref{eq:saddle}. 

The choice of stepsizes is crucial for the convergence of an optimization algorithm. 
AMP can be interpreted as a special case of PDHG with an adaptive stepsize updating schedule~\cite{rangan2019convergence}. Specifically,
let 
\begin{equation}
\tau_s^{t} = \frac{1}{\tau_x^{t}-1},\quad \tau_x^{t+1} =1+\frac{d^t}{c}\tau_x^t,
\label{eq:def_tau_s}
\end{equation}
where $d^t=\|\xbf^{t+1}\|_0/N$ with $\|\xbf^{t+1}\|_0$ (the $\ell_0$ quasi-norm) denoting the number of nonzero coordinates 
of $\xbf^{t+1}$. By \eqref{eq:def_fg} and \eqref{eq:def_conv_conj}, we have $g^*(\sbf)=\langle \ybf,\sbf\rangle+\frac{1}{2}\|\sbf\|_2^2$. For easy comparison, we use the same notation for the soft-thresholding function (proximal operator for the $\ell_1$-norm) as in \cite{BayatiMontanari2012}. 
For any $\theta>0$, $\ubf\in\mathbb{R}^N$, define
\begin{equation}
\eta\left(\ubf;\theta\right) :=\argmin_{\xbf\in\mathbb{R}^N} \|\xbf\|_1 + \frac{1}{2\theta}\|\xbf - \ubf\|_2^2.
\label{eq:def_eta}
\end{equation}
Then \eqref{eq:primal_dual_general} can be written as
\begin{equation}
\begin{split}
\sbf^{t+1} &=\frac{1}{\tau_x^t}\left(\Abf\xbf^t - \ybf\right) + \left(1-\frac{1}{\tau_x^t}\right)\sbf^t,\\
\xbf^{t+1} &=\eta\left(\xbf^t -\tau_x^t \Abf^T \sbf^{t+1};\gamma\tau_x^t\right).
\end{split}
\label{eq:AMP1}
\end{equation}
Let $\zbf^{t}:=-\tau_x^{t}\sbf^{t+1}$ for all $t\geq 1$ and notice from \eqref{eq:def_tau_s} that $(\tau_x^t-1)/\tau_x^{t-1}=d^{t-1}/c$. 
We can see that \eqref{eq:AMP1} matches the AMP algorithm (see \cite{Donoho.etal2009} and \cite{BayatiMontanari2012}), but with a different choice for the threshold of the soft-thresholding function.
We emphasize that the choice of the threshold in \cite{BayatiMontanari2012} does not guarantee that AMP will converge to the LASSO solution for finite dimensional problems, whereas the choice in \eqref{eq:AMP1} guarantees that whenever \eqref{eq:AMP1} converges, it converges to the LASSO solution for arbitrary finite dimensional $\Abf$.

In many optimization algorithms, the stepsize depends on 
$\sigma_\text{max}(\Abf)$,
the largest singular value of $\Abf$, whereas the algorithm defined in \eqref{eq:primal_dual_general} with $\tau_s^t$ and $\tau_x^t$ updated according to \eqref{eq:def_tau_s}, which is equivalent to AMP \eqref{eq:AMP1}, does not depend on
$\sigma_\text{max}(\Abf)$. Therefore, in order to have an algorithm that is more robust than AMP with arbitrary finite dimensional $\Abf$ while retaining the fast convergence of AMP, we introduce a parameter $e$ to \eqref{eq:def_tau_s} that depends on 
$\sigma_\text{max}(\Abf)$ while keeping the general structure of the updating schedule in \eqref{eq:def_tau_s}. Specifically, choosing $0 < e < \min\{1,4/(\sigma_\text{max}^2(\Abf) + 2)\}$, we modify \eqref{eq:def_tau_s} as
\begin{equation}
\tau_s^{t} = \frac{e}{\tau_x^{t}-e},\quad \tau_x^{t+1} =1+\frac{d^t}{c}\tau_x^t.
\label{eq:def_tau_s_new}
\end{equation}
Such a choice of $e$ ensures local stability of our proposed algorithm
(see Section \ref{sec:local}).
Our proposed AMP variant is \eqref{eq:primal_dual_general} with the stepsize updating schedule defined in \eqref{eq:def_tau_s_new}. Similar to the derivation of \eqref{eq:AMP1} from \eqref{eq:primal_dual_general} and \eqref{eq:def_tau_s}, we can write the proposed AMP variant as in Algorithm~\ref{algo:amp_variant}.

\begin{algorithm}
\caption{Proposed AMP Variant}
\label{algo:amp_variant}
\textbf{Input:} $\Abf$, $\ybf$, $0<e<\min\{1,4/(\sigma_\text{max}^2(\Abf) + 2)\}$, $t_{\text{max}}$\\
\textbf{Initialization:} $\xbf^0$, $\sbf^0$, $\tau_x^0$
\begin{algorithmic}
\For{$0\leq t \leq t_{\text{max}}$}
\begin{equation}
\begin{split}
\sbf^{t+1} &=\frac{e}{\tau_x^t}\left(\Abf\xbf^t - \ybf\right) + \left(1-\frac{e}{\tau_x^t}\right)\sbf^t\\
\xbf^{t+1} &=\eta\left(\xbf^t -\tau_x^t \Abf^T \sbf^{t+1};\gamma\tau_x^t\right)\\
\tau_x^{t+1} &=1+\frac{d^t}{c}\tau_x^t \quad\text{with } d^t=\|\xbf^{t+1}\|_0/N
\end{split}
\label{eq:our_algo}
\end{equation}
\EndFor
\end{algorithmic}
\textbf{Output:} $\xbf^{t_{\text{max}}}$
\end{algorithm} 
\section{Local Stability Analysis}
\label{sec:local}

We now study the local stability of Algorithm~\ref{algo:amp_variant} around its equilibrium point, which is the LASSO solution. We first discuss conditions under which our analysis is valid, and then show that Algorithm~\ref{algo:amp_variant} is locally stable under these conditions.

\subsection{Assumptions} 
\label{subsec:assumptions}

When $\Null(\Abf)$, the null space of $\Abf$, contains nonzero components, the objective function is not strictly convex in $\xbf$, hence there may be multiple solutions. Conditions for the uniqueness of the LASSO solution have been studied by Tibshirani~\cite{Tibshirani2013}, which states that a sufficient~\cite[Lemma 2]{Tibshirani2013} and necessary~\cite[Lemma 16]{Tibshirani2013} condition for $\eqref{eq:lasso}$ to admit a unique solution is that 
$\Null(\col_{\mathcal{E}}(\Abf))=\{\mathbf{0}\}$, where $\mathcal{E}=\{i\in\{1,\ldots,N\}\,\vert\, |[\Abf^T(\ybf-\Abf\xbf^*)]_i| = \lambda\}$ and $\col_{\mathcal{E}}(\Abf)$ is the submatrix of $\Abf$ formed by deleting the $i^{th}$ column of $\Abf$ for all $i\not\in \mathcal{E}$. 
By the first order optimality condition for~\eqref{eq:lasso}, we have $\Abf^T(\ybf-\Abf\xbf^*)=\lambda \vbf$ with $\vbf\in\partial_{\xbf}(\|\xbf\|_1)\vert_{\xbf=\xbf^*}$. Note that $v_i=\text{sign}(x_i^*)$ if $x_i^*\neq 0$ and $v_i\in[-1,1]$ if $x_i^*=0$.
Let $K=\{i\in\{1,\ldots,N\}\vert x^*_i\neq 0\}$, thus $K\subset\mathcal{E}$.
Then the necessary condition implies that when the LASSO solution is unique, we have that $|K|\leq \min \{n,N\}$; this is a condition that we need to prove our result.
A more explicit sufficient condition for uniqueness in the almost sure sense is also provided in~\cite[Lemma 4]{Tibshirani2013}: if entries of $\Abf$ are drawn from a continuous probability distribution on $\mathbb{R}^{n\times N}$, then the LASSO solution is unique with probability one regardless of the dimension of $\Abf$.

Note that with $\eta(\cdot)$ being the soft-thresholding function \eqref{eq:def_eta}, the definition of $d^t$ in \eqref{eq:our_algo} can be written as
\begin{equation*}
d^t = \left\vert  \{  i\in\{1,\ldots,N\} : \vert   [\xbf^t -\tau_x^t \Abf^T \sbf^{t+1}]_i \vert > \gamma\tau_x^t   \} \right\vert,
\end{equation*}
where we can further replace $\sbf^{t+1}$ by $\frac{e}{\tau_x^t}\left(\Abf\xbf^t - \ybf\right) + \left(1-\frac{e}{\tau_x^t}\right)\sbf^t$, so that $d^t$ only depends on the iterates at the $t^{th}$ iteration. Similarly, the update of $\xbf^{t+1}$ in \eqref{eq:our_algo} can also be written as a function of iterates at the $t^{th}$ iteration only. Therefore, letting $\vbf^t\in\mathbb{R}^{n+N+1}$ be defined as $\vbf^t:=[\sbf^t; \xbf^t; \tau_x^t]$ for all $t\geq 0$, we have that \eqref{eq:our_algo} defines a nonlinear operator $G:\mathbb{R}^{n+N+1}\to\mathbb{R}^{n+N+1}$ such that $\vbf^{t+1}=G(\vbf^t)$. Note that $G$ is differentiable at $[\sbf^t; \xbf^t; \tau_x^t]$ except when there exists an $i\in\{1,\ldots, N\}$, such that
\begin{equation*}
[\xbf^t - \Abf^T\left(e(\Abf\xbf^t-\ybf) + (\tau_x^t - e)\sbf^t\right)]_i = \pm\gamma\tau_x^t,
\end{equation*}
which has probability zero if $\Abf$ obeys a continuous distribution.
To summarize, we make the following two assumptions on components in \eqref{eq:lasso}: 
\begin{enumerate}
\item The matrix $\Abf$ is drawn from a continuous probability distribution on $\mathbb{R}^{n \times N}$.
\item The regularization parameter $\gamma>0$.
\end{enumerate}

\subsection{Stability around the Equilibrium Point}
\label{subsec:local_analysis}

Having clarified our assumptions, we now prove our main result, which is the local stability guarantee of Algorithm~\ref{algo:amp_variant} as stated in the following proposition.
\begin{prop}
\label{prop:local_stability}
Consider the LASSO problem defined in \eqref{eq:lasso} and suppose that the conditions of Section~\ref{subsec:assumptions} are satisfied. Then Algorithm~\ref{algo:amp_variant} is stable around its equilibrium point with probability one.
\end{prop}

\begin{proof}
Suppose that $G$ is differentiable around the equilibrium point $\widehat{\vbf}$. Then the local stability of $G$ around $\widehat{\vbf}$ is determined by the largest eigenvalue (in modulus)  of the Jacobian matrix $\Jbf$ of $G$ evaluated at $\widehat{\vbf}$. The expression for $\Jbf$ is
\begin{equation}
\Jbf=\left[\begin{matrix}
\left(1- e/\widehat{\tau}_x\right)\Ibf_{n}                          & (e/\widehat{\tau}_x)\Abf                                           &  \mathbf{0}_{n\times 1}\\
\left(e-\widehat{\tau}_x\right) \widehat{\Dbf}\Abf^T    &\widehat{\Dbf}\left(\Ibf_{N}- e\Abf^T\Abf\right)   & -\widehat{\Dbf}\Abf^T\widehat{\sbf}\\
\mathbf{0}_{1\times n}                                                          & \mathbf{0}_{1\times N}                                              & \widehat{d}/c
\end{matrix}\right],
\label{eq:jacobian}
\end{equation}
where $\widehat{\Dbf}$ is a diagonal matrix defined as $\widehat{D}_{ii}=\mathbb{I}\{[\widehat{\xbf}]_i\neq 0\}$,
$\mathbb{I}$ is the indicator function, $\Ibf_n$ is the $n\times n$ identity matrix, and $\mathbf{0}_{n\times m}$ is an $n\times m$ zeros matrix.
Below, we show that with an appropriate choice of $e$, the eigenvalue of $\Jbf$ with the largest modulus is within the unit circle of the complex plane. 

Let $\alpha=1-e/\widehat{\tau}_x$ and let $\vert \Abf \vert$ denote the determinant of a matrix $\Abf$. Then
\begin{align}
&\left\vert \Jbf\!-\!\lambda \Ibf \right\vert \!\overset{(a)}{=}\!\left(\widehat{d}/c - \lambda \right)\! \left\vert\begin{matrix}
\left(\alpha -\lambda \right)\Ibf_{n}                                   &(e/\widehat{\tau}_x)\Abf\\
\left(e-\widehat{\tau}_x\right) \widehat{\Dbf}\Abf^T    & \!\!\widehat{\Dbf}-\lambda \Ibf_{N}- e\widehat{\Dbf} \Abf^T\Abf 
\end{matrix}   \right\vert\nonumber\\
&\overset{(b)}{=}\!\left(\widehat{d}/c- \lambda \right)\! \left\vert\begin{matrix}
\left(\alpha -\lambda \right)\Ibf_{n}              & (e/\widehat{\tau}_x)\Abf\\
\mathbf{0}_{N\times n}                                                     & \widehat{\Dbf}-\lambda\Ibf_{N}+ \lambda e/(\alpha - \lambda) \widehat{\Dbf}\Abf^T\Abf
\end{matrix}   \right\vert\nonumber\\
&= \left( \widehat{d}/c- \lambda \right)(\alpha-\lambda)^{n-N}\nonumber\\
&\qquad\qquad \cdot\left\vert (\alpha-\lambda)\widehat{\Dbf}-(\alpha-\lambda)\lambda \Ibf_{N}+ \lambda e \widehat{\Dbf}\Abf^T\Abf\right\vert,
\label{eq:mainEq}
\end{align}
where step $(a)$ follows by expanding the last row of $\Jbf-\lambda \Ibf$, and step $(b)$ follows by subtracting $\frac{e-\widehat{\tau}_x}{\alpha-\lambda}\widehat{\Dbf}\Abf^T$ times the first row from the second row and noticing that 
{\footnotesize{ 
\[
e\!+\!\frac{e}{\widehat{\tau}_x}\frac{e-\widehat{\tau}_x}{\alpha-\lambda}
=
e\!\left(1\!+\!\frac{1}{\alpha-\lambda}\left(\frac{e}{\widehat{\tau}_x}\!-\!1\right)\right)
=
e\!\left(1\!-\!\frac{\alpha}{\alpha-\lambda}\right)=\frac{-\lambda e}{\alpha-\lambda}.
\]
}}
To calculate $\left\vert (\alpha-\lambda)\widehat{\Dbf}-(\alpha-\lambda)\lambda \Ibf_{N}+ \lambda e \widehat{\Dbf}\Abf^T\Abf\right\vert$, we first introduce some notation. For a matrix $\Bbf\in\mathbb{R}^{N\times N}$ and index set $K\subset \{1,2,...,N\}$, let $[\Bbf]_K\in\mathbb{R}^{|K|\times |K|}$ denote the submatrix of $\Bbf$ obtained by eliminating the $i^{th}$ row and $i^{th}$ column of $\Bbf$ for all $i\not\in K$. Moreover, let $\row_K(\Bbf)$ (resp. $\col_K(\Bbf)$) denote the submatrix of $\Bbf$ formed by deleting the $i^{th}$ row (resp. column) of $\Bbf$ for all $i\not\in K$.

Letting $\Bbf=(\alpha-\lambda)\widehat{\Dbf}-(\alpha-\lambda)\lambda\Ibf_{N}+ \lambda e \widehat{\Dbf}\Abf^T\Abf$, we have 
\begin{equation*}
\row_{\{i\}}(\Bbf)=\begin{cases}
-\lambda(\alpha-\lambda)\ebf_i^T,&\text{if }i\in K^c\\
(1-\lambda)(\alpha-\lambda)\ebf_i^T + e\lambda \row_{\{i\}}(\Abf^T\Abf), &\text{if }i\in K
\end{cases},
\end{equation*}
where all but the $i^{th}$ coordinates of $\ebf_i\in\mathbb{R}^N$ are zero and the $i^{th}$ coordinate is 1, and $K=\{i\in\{1,...,N\}\vert \widehat{D}_{ii} = 1\}$.
By expanding the $i^{th}$ row of $\Bbf$ for all $i\not\in K$, we have
\begin{equation*}
\left\vert \Bbf\right\vert \!=\! (-\lambda(\alpha - \lambda))^{N(1-\widehat{d})}\!\left\vert(\alpha-\lambda)(1-\lambda)\Ibf_{N\widehat{d}} + \lambda e [\Abf^T\Abf]_K\right\vert.
\end{equation*}
Plugging the above into \eqref{eq:mainEq}, we have
\begin{align}
\left\vert \Jbf -\lambda \Ibf \right\vert &= \left( \widehat{d}/c - \lambda \right)(\alpha-\lambda)^{n-N}(-\lambda(\alpha - \lambda))^{N(1-\widehat{d})}\nonumber\\
&\qquad\cdot\left\vert(\alpha-\lambda)(1-\lambda)\Ibf_{N\widehat{d}} + \lambda e [\Abf^T\Abf]_K\right\vert\nonumber\\
&=(-1)^{N(1-\widehat{d})}\left(\widehat{d}/c - \lambda \right)\lambda^{N(1-\widehat{d})}(\alpha-\lambda)^{n-N\widehat{d}}\nonumber\\
&\qquad\cdot\left\vert(\alpha-\lambda)(1-\lambda)\Ibf_{N\widehat{d}} + \lambda e [\Abf^T\Abf]_K\right\vert.\label{eq:roots1}
\end{align}

Let $\Hbf=[\Abf^T\Abf]_K$, we now need to solve for $\lambda$ in the following equation:
\begin{equation}
\left\vert (\alpha-\lambda)(1-\lambda)\Ibf_{N\widehat{d}} + \lambda e \Hbf \right\vert=0. 
\label{eq:rest_eig}
\end{equation}
First, we check that $\lambda=0$ is not a solution to \eqref{eq:rest_eig}. Plugging $\lambda=0$ into \eqref{eq:rest_eig}, we have 
\begin{equation*}
\left\vert (\alpha-\lambda)(1-\lambda)\Ibf_{N\widehat{d}} + \lambda e \Hbf \right\vert=|\alpha \Ibf_{N\widehat{d}}|= \left(1-e/\widehat{\tau}_x\right)^{N\widehat{d}}>0,
\end{equation*}
where the last inequality holds because $e\in (0,1]$ and $\widehat{\tau}_x>1$, hence $1-e/\widehat{\tau}_x>0$. Now that $\lambda\neq 0$, we divide both sides of \eqref{eq:rest_eig} by $(e\lambda)^{N\widehat{d}}$:
\begin{equation}
\left\vert \frac{(\alpha-\lambda)(1-\lambda)}{\lambda e}\Ibf_{N\widehat{d}} +  \Hbf \right\vert=0.
\label{eq:roots2}
\end{equation}
Therefore, $\lambda$ is a solution to \eqref{eq:rest_eig} if and only if $-\frac{(\alpha-\lambda)(1-\lambda)}{\lambda e}$ is an eigenvalue of $\Hbf$. Let $\text{sp}(\Hbf)$ denote the spectrum (i.e., set of eigenvalues) of $\Hbf$, and define
\begin{equation}
m:=\min_{\lambda\in\text{sp}(\Hbf)} \lambda,\qquad\text{and}\qquad M:=\max_{\lambda\in\text{sp}(\Hbf)} \lambda.
\label{eq:def_m_M}
\end{equation}
Note that $\Hbf=[\Abf^T\Abf]_K=(\col_K(\Abf))^T\col_K(\Abf)$. By condition of the uniqueness of LASSO solution, we have that $\col_K(\Abf)$ is non-singular, hence $m>0$. Let $u=-\frac{(\alpha-\lambda)(1-\lambda)}{\lambda e}$, then
\begin{equation}
\lambda^2 + (eu-\alpha - 1)\lambda + \alpha = 0.
\label{eq:lambda_quad}
\end{equation}
Let $b:=eu -\alpha - 1$, and we solve \eqref{eq:lambda_quad} for $\lambda$.
When $|b|<2\sqrt{\alpha}$, we have complex roots $\lambda=\frac{1}{2}(-b\pm i\sqrt{4\alpha - b^2})$, hence $|\lambda|=\frac{1}{2}(\sqrt{b^2 + 4\alpha-b^2})=\sqrt{\alpha}<1$. 
When $|b|\geq 2\sqrt{\alpha}$, we have real roots. Define
\begin{equation}
\begin{split}
h_1(b)&:=\frac{1}{2}(-b+\sqrt{b^2-4\alpha}),\\
h_2(b)&:=\frac{1}{2}(-b-\sqrt{b^2-4\alpha}).
\end{split}
\label{eq:def_h1h2}
\end{equation}
Notice that 
\begin{equation}
\begin{split}
h_1'(b) &= \frac{1}{2}\left(-1 + \frac{b}{\sqrt{b^2-4\alpha}}\right)
\begin{cases}
> 0,&\text{if }b>2\sqrt{\alpha}\\
< 0,&\text{if }b\leq -2\sqrt{\alpha}
\end{cases},\\
h_2'(b) &= \frac{1}{2}\left(-1 - \frac{b}{\sqrt{b^2-4\alpha}}\right)
\begin{cases}
< 0,&\text{if }b\geq 2\sqrt{\alpha}\\
> 0,&\text{if }b<-2\sqrt{\alpha}
\end{cases}.
\end{split}
\label{eq:h1h2_derivative}
\end{equation}
Also notice that $h_1(b)\uparrow 0$ as $b\uparrow\infty$, and $h_2(b)\downarrow 0$ as $b\downarrow -\infty$. The graph of $h_1$ and $h_2$ as a function of $b$, respectively, is shown in Fig. \ref{fig:graph_lambda_real}. It can be seen that
\begin{equation*}
\max(|h_1(b)|,|h_2(b)|)=\begin{cases}
h_1(b),&\text{if } b\leq -2\sqrt{\alpha}\\
-h_2(b), &\text{if } b\geq 2\sqrt{\alpha}
\end{cases}.
\end{equation*}

\begin{figure}[t!]
\centering
\includegraphics[width=0.45\textwidth]{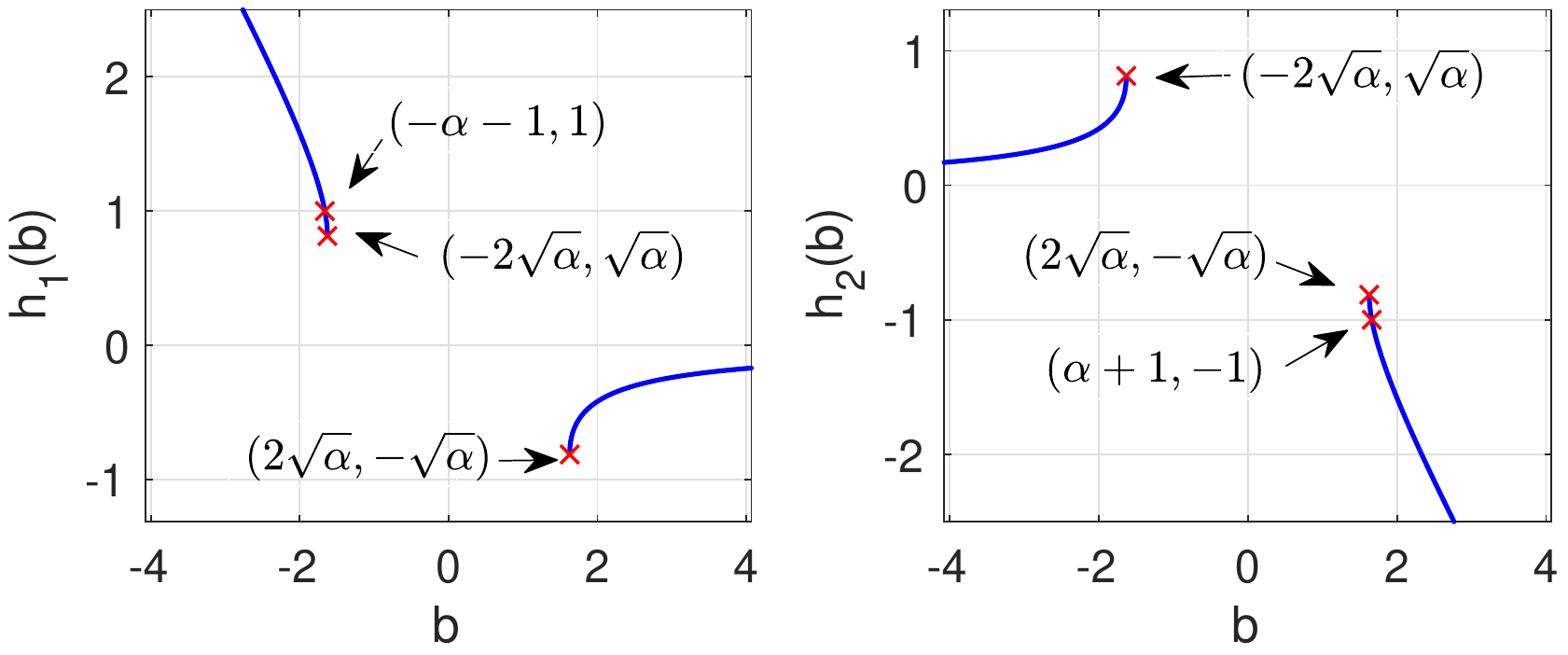}
\caption{Graph of $h_1$ and $h_2$ as a function of $b$, where a specific $e$ is chosen for demonstration purposes.}
\label{fig:graph_lambda_real}
\end{figure}

First consider $b\leq-2\sqrt{\alpha}$. Recall that $b=eu-\alpha-1>em-\alpha-1>-\alpha-1$, since $e>0$ and $m>0$. Notice from \eqref{eq:def_h1h2} and \eqref{eq:h1h2_derivative} that $h_1(-\alpha-1)=1$, $h_1(-2\sqrt{\alpha})=\sqrt{\alpha}<1$, and that $h_1(b)$ is monotone decreasing when $b \leq -2\sqrt{\alpha}$. Therefore, we have that $|h_1(b)|<1$, when $b\leq-2\sqrt{\alpha}$.

Next consider $b \geq 2\sqrt{\alpha}$. 
Notice that $h_2(1+\alpha)=-1$ and that $h_2(b)$ is monotone decreasing when $b\geq 2\sqrt{\alpha}$. 
Therefore, in order to have $|h_2(b)|<1$, we need $b=eu-\alpha-1<1+\alpha, \forall u\in\text{sp}(\Hbf)$. This condition is satisfied if we let $e<\frac{2(1+\alpha)}{M}$.  Recall that $\alpha=1-e/\widehat{\tau}_x$. By \eqref{eq:our_algo}, we have $1/\widehat{\tau}_x = 1- \widehat{d}/c$. Combining the above analysis on $h_1$ and $h_2$, it follows that the condition on the parameter $e$ for all eigenvalues of $\Jbf$ to be within the unit circle of the complex plane is 
\begin{equation}
 \quad 0<e<\min\left\{1,\frac{4}{M+ 2(1-\widehat{d}/c)}\right\}.
\label{eq:e_bound}
\end{equation}

The upper-bound for $e$ in \eqref{eq:e_bound} is tight. However, it depends on the equilibrium point. It is desirable to have a condition on $e$ that does not depend on knowledge about the equilibrium point, so that an appropriate value for $e$ can be chosen before running the algorithm \eqref{eq:our_algo}. Notice that $\Hbf=[\Abf^T\Abf]_K$ is a principal submatrix of the symmetric matrix $\Abf^T\Abf$. The interlacing property of eigenvalues implies that $M\leq L:=\max_{\lambda\in\text{sp}(\Abf^T\Abf)}\lambda$. Moreover, the condition for the uniqueness of the LASSO solution implies that $\widehat{d}/c \in [0,1]$.
Therefore, local stability is guaranteed if
\begin{equation}
0< e \leq \min\{1, 4/(L+2)\}.
\label{eq:e_bound_weaker}
\end{equation}
\end{proof}

\subsection{Random Matrices with i.i.d. Entries} 
In the special case where the matrix $\Abf \in\mathbb{R}^{n\times N}$ is the upper left corner of a doubly infinite array\footnote{That is, we have an array $\{X_{ij}\}, i=1,2,\ldots; j=1,2,\ldots$ and 
$\Abf=(X_{ij}), i=1,2,\ldots, n; j=1,2,\ldots, N$.} of i.i.d. zero-mean random variables with finite fourth moment and is normalized such that the variance is $1/n$, the asymptotic largest singular value of $\col_K(\Abf)\in\mathbb{R}^{n\times N\widehat{d}}$, where $N\widehat{d}<n$, is $1+\sqrt{N\widehat{d}/n}=1+\sqrt{\widehat{d}/c}$ with probability one~\cite{yin1988limit}. It follows that the denominator of the upper-bound in \eqref{eq:e_bound} is
\begin{equation*}
M+2(1-\widehat{d}/c) \!=\! 1+\widehat{d}/c+2\sqrt{\widehat{d}/c} + 2 - 2\widehat{d}/c=3+2\sqrt{\widehat{d}/c}-\widehat{d}/c.
\end{equation*}
Let $x=\widehat{d}/c$, hence $x\in [0,1]$. Define $f(x)=2\sqrt{x}-x$ and notice that $f$ is monotone increasing on $[0,1]$. Therefore, $f(x)<f(1)=1$, which implies that $M+2(1-\widehat{d}/c)\leq 4$. That is, local stability for large zero-mean random matrices with variance $1/n$ is guaranteed by setting $e=1$, which, as mentioned before, makes Algorithm~\ref{algo:amp_variant} coincide with the original AMP \eqref{eq:AMP1}, as seen in \cite{Donoho.etal2009} and \cite{BayatiMontanari2012}.

\section{Numerical Demonstration}
\label{sec:numerical}

To demonstrate the efficiency of our proposed AMP variant, we compare it with the original AMP that uses the calibration method proposed in \cite{BayatiMontanari2012}, a PDHG algorithm with a fixed stepsize that guarantees convergence (see \cite{HeYuan2012}), and a popular convex optimization algorithm, fast iterative shrinkage and thresholding algorithm (FISTA)~\cite{BeckTeboulle2009}. Because our proposed algorithm is inspired by AMP and depends on a parameter $e$, we call it eAMP and choose $e$ as the upper bound in \eqref{eq:e_bound_weaker}.
In addition, we include results for eAMP with $e=1$, which is of the same form as AMP but the threshold for the thresholding function at each iteration is different from that in \cite{BayatiMontanari2012}.

In all the simulations, the problem dimension is 
$N=2000$, $n=1000$. The data vector $\ybf$ is obtained by $\ybf=\Abf \xbf_0 + \wbf$, where entries of $\wbf$ are independent realizations of a Gaussian distribution with mean zero and variance $\sigma_w^2$ and entries of $\xbf_0$ are independent realizations of a Bernoulli(0.1)-Uniform[-1,1] distribution (i.e., $X_0=BU$ with $B\sim\text{Bernoulli}(0.1)$ and $U\sim\text{Uniform}[-1,1]$). The value of $\sigma_w^2$ satisfies $10\log_{10}\left(\frac{\|\Abf\xbf_0\|_2^2}{n\sigma_w^2}\right)=25$, i.e., the signal-to-noise ratio is 25dB. All tested algorithms are initialized with an all-zero vector. Since FISTA and PDHG have theoretical convergence guarantees, we present their results only for comparison of empirical convergence speed, hence we sometimes stop them early when the convergence speed comparison is clear. 

\begin{figure}
\centering
\begin{subfigure}{0.45\textwidth}
\centering
\includegraphics[width=\textwidth]{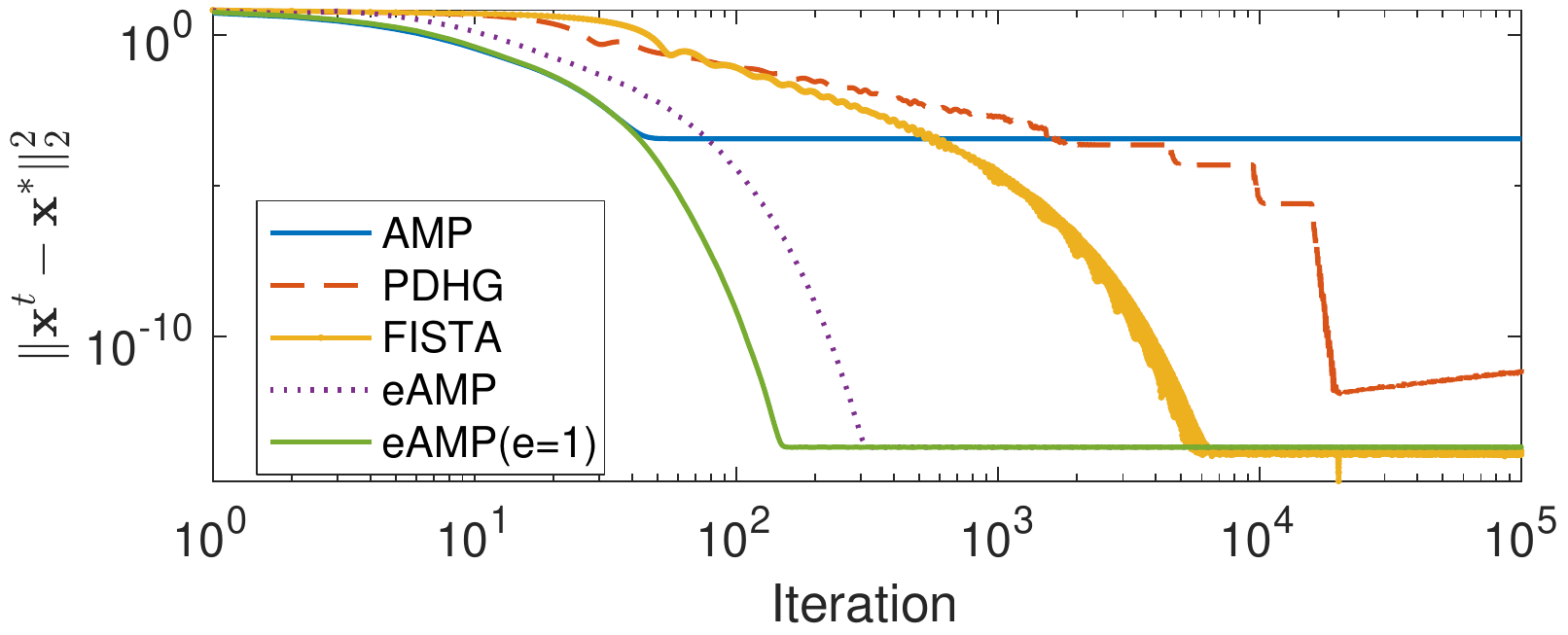}
\caption{Entries of the matrix $\Abf$ are drawn independently from a zero-mean Gaussian distribution.}
\label{fig:sim_standard}
\end{subfigure}
\begin{subfigure}{0.45\textwidth}
\centering
\includegraphics[width=\textwidth]{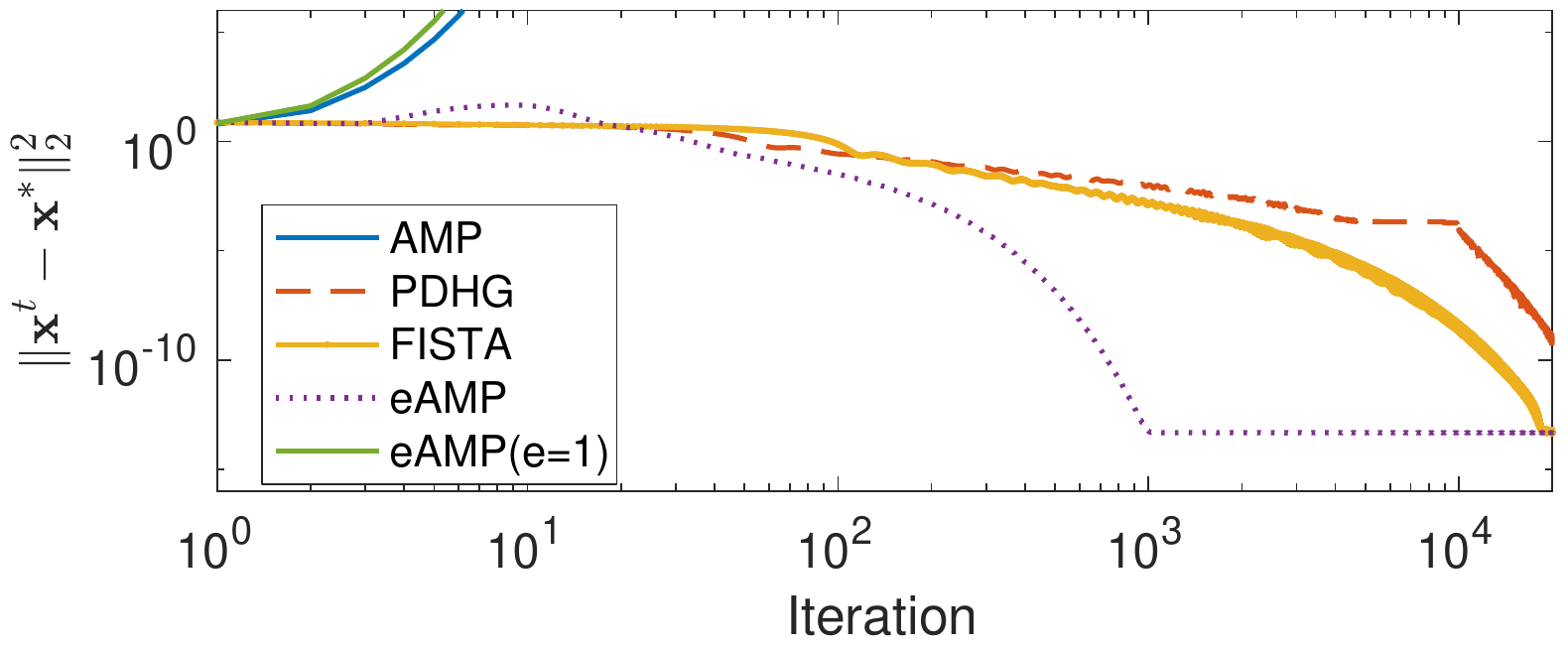}
\caption{Rows of the matrix $\Abf$ are drawn independently from a zero-mean multivariate Gaussian distribution.}
\label{fig:sim_correlated}
\end{subfigure}
\caption{Comparison of different algorithms for solving the LASSO problem \eqref{eq:lasso}.}
\end{figure}

For the first set of simulations, we use matrices $\Abf$ whose entries are i.i.d.\ zero-mean Gaussian,
which is the case studied in \cite{BayatiMontanari2012} in the limit as $N,n\to\infty$. The simulation results are shown in Fig. \ref{fig:sim_standard}. We notice that while AMP seems to have converged, it does not converge to the LASSO solution $\xbf^*$, whereas eAMP with both choices of $e$ has converged to the LASSO solution. Moreover, the empirical convergence speed (in terms of number of iterations) of eAMP is much faster than that of FISTA or PDHG with our choice of $e$, though a smaller $e$ may lead to a slower convergence.

For the second set of simulations, we use matrices $\Abf$ whose rows are independent realizations of a zero-mean multivariate Gaussian distribution, where diagonal entries of the covariance matrix have value $1/n$ and off-diagonal entries have value $0.01/n$.
The simulation results are shown in Fig. \ref{fig:sim_correlated}. In this case, AMP and eAMP with the inappropriate choice of $e=1$ have diverged, whereas eAMP with $e$ as the upper bound in
\eqref{eq:e_bound_weaker} has converged to the LASSO solution and requires far fewer iterations than FISTA and PDHG.

While our analysis in Section \ref{sec:local} only guarantees local stability for eAMP, the encouraging simulation results suggest that a global convergence result might be possible; we leave the global convergence analysis for future work. 
\section{Conclusion}
\label{sec:conclusion}
In this paper, we proposed an AMP variant (Algorithm \ref{algo:amp_variant})  for solving the LASSO problem \eqref{eq:lasso}. 
Unlike the work in \cite{BayatiMontanari2012} that analyzes the limiting behavior of AMP iterates as the iteration number goes to infinity for \emph{infinite} dimensional problems, we focused on \emph{finite} dimensional problems. 
Specifically, for any finite dimensional matrix $\Abf$, whenever our algorithm converges, it converges to the LASSO solution. This is not the case for AMP with finite dimensional $\Abf$ even when $\Abf$ has i.i.d. Gaussian entries, as shown in Fig. \ref{fig:sim_standard}. 
The proposed algorithm contains a parameter $e$ that depends on the largest singular value of $\Abf$. In Proposition \ref{prop:local_stability}, we provided conditions on $e$ under which the algorithm is locally stable around the LASSO solution with probability one when entries of $\Abf$ are drawn from a continuous distribution.
Finally, numerical results showed that the number of iterations required for our algorithm to converge is orders of magnitude smaller than optimization algorithms such as FISTA~\cite{BeckTeboulle2009} and PDHG~\cite{HeYuan2012}.


\bibliographystyle{IEEEtran}
\bibliography{optimization}
\end{document}